\documentclass[letterpaper, 10 pt, conference]{ieeeconf}  %

\IEEEoverridecommandlockouts                              %
\usepackage[cmex10]{amsmath}
\usepackage{amsfonts,amssymb,bm}

\usepackage{graphicx,color,psfrag}		%
\usepackage{cite}			%
\usepackage{subfigure}
\usepackage[hidelinks]{hyperref}
\let\labelindent\relax 		%
\usepackage{enumitem}

\usepackage{epstopdf}	%

\usepackage[capitalise]{cleveref}
\usepackage{multirow}
\usepackage{booktabs}
\crefname{equation}{}{} %
\crefname{section}{Sec.}{Sec.}
\newcommand{\domf}{\mathcal{F}}
\newcommand{\ab}{($a, b$)}    %

\usepackage{algorithm}
\usepackage{algpseudocode}
\makeatletter
\def\BState{\State\hskip-\ALG@thistlm}
\algrenewcommand\ALG@beginalgorithmic{\footnotesize}
\makeatother

\graphicspath{{.},{figures/}}

\newcommand{\transpose}{\text{T}}
\newcommand{\transp}{\transpose}

\newcommand{\dimx}{{n_\text{x}}}
\newcommand{\dimu}{{n_\text{u}}}
\newcommand{\nfeat}{m}     %
\newcommand{\alow}{a_\text{min}}
\newcommand{\aup}{a_\text{max}}
\newcommand{\blow}{b_\text{min}}
\newcommand{\bup}{b_\text{max}}
\newcommand{\kSE}{k_{\text{SE}}}
\newcommand{\kLQR}{k_{\text{pLQR}}}
\newcommand{\kLQRmulti}{k_{\text{pLQR},\nfeat}}
\newcommand{\kLQRtwo}{k_{\text{nLQR}}}
\newcommand{\GP}{\mathcal{GP}}
\newcommand{\setstable}{\mathcal{F}}    %

\providecommand{\abs}[1]{\lvert#1\rvert}

\newcommand{\Dc}{\mathcal{D}}

\newcommand{\Nc}{\mathcal{N}}

\newcommand{\field}[1]{\mathbb{#1}}
\newcommand{\R}{\field{R}}

\DeclareMathOperator*{\E}{\field{E}}

\DeclareMathOperator*{\V}{\field{V}}

\newtheorem{definition}{Definition}
\newtheorem{proposition}{Proposition}

\newtheorem{fact}{Fact}

\newtheorem{example}{Example}

\newcommand{\ie}{i\/.\/e\/.,\/~}
\newcommand{\eg}{e\/.\/g\/.,\/~}

\newcommand{\fig}{Fig\/.\/~}

\newcommand{\sect}{Sec\/.\/~}

\newcommand{\The}{Theorem~}

\newcommand{\Ex}{Example~}

\newcommand{\etal}{\emph{et al}\/.\/~}

\title{\LARGE \bf
On the Design of LQR Kernels for Efficient Controller Learning
}

\author{Alonso Marco$^{1}$, Philipp Hennig$^{1}$, Stefan Schaal$^{1,2}$ and Sebastian Trimpe$^{1}$%
\thanks{$^{1}$Max Planck Institute for Intelligent Systems, 
	72076 T\"ubingen, Germany. \{amarco, phennig, sschaal, strimpe\}@tuebingen.mpg.de}%
\thanks{$^{2}$Computational Learning and Motor Control Lab, University of Southern California, Los Angeles, USA.}%
\thanks{This work was supported in part by the Max Planck Society, the Max Planck ETH Center for Learning Systems, National Science Foundation grants IIS-1205249, IIS-1017134, EECS-0926052, the Office of Naval Research, and the Okawa Foundation.}%
}

\usepackage{fancyhdr}
\newcommand{\mytitle}{\textbf{Accepted final version.}
To appear in \textit{56th IEEE Conference on Decision and Control}.\\
\copyright 2017 IEEE. Personal use of this material is permitted. Permission from IEEE must be obtained for all other uses, in any current or future media, including reprinting/republishing this material for advertising or promotional purposes, creating new collective works, for resale or redistribution to servers or lists, or reuse of any copyrighted component of this work in other works.}
\begin{document}

\maketitle

\thispagestyle{fancy}	
\pagestyle{empty}

%
%
\begin{abstract}
Finding optimal feedback controllers for 
nonlinear dynamic systems from data is hard. 
Recently, Bayesian optimization (BO) has been proposed as a powerful framework for direct controller 
tuning from experimental trials.   
For selecting the next query point and finding the global optimum,
BO relies on a probabilistic description of the latent objective function, typically a Gaussian process (GP).
As is shown herein, GPs with a common kernel choice can, however, lead to poor learning outcomes on standard quadratic control problems.
For a first-order system, we construct two kernels that specifically leverage the structure of the well-known Linear Quadratic Regulator (LQR), yet retain the flexibility of Bayesian nonparametric learning.  Simulations of uncertain linear and nonlinear systems demonstrate that the LQR kernels yield superior 
learning performance.
\end{abstract}

\section{Introduction}
\label{sec:intro}
A core problem of learning control is to determine optimal feedback controllers for (partially) unknown nonlinear systems from experimental data.  
Reinforcement learning (RL) \cite{SuBa98,KoBaPe13} is a promising framework for this, yet often requires performing many experiments on the physical system to even find suitable controllers, which limits the applicability of such techniques.  Therefore, a lot of research effort has been invested into data efficiency of RL aiming at learning controllers from as few experiments as possible. 
Recently,  Bayesian optimization (BO) has been proposed for RL as a promising approach in this direction.
BO employs a probabilistic description of the latent objective function (typically a Gaussian process (GP)), which allows for selecting next control experiments in a principled manner, \eg to maximize information gain \cite{MaHeBoScTr16} or perform safe exploration \cite{BeScKr16}.

While BO provides a promising framework for learning controllers in fairly general settings, the full power of Bayesian learning is often not exploited.
A key advantage of Bayesian methods is that they allow for combining prior problem knowledge with learning from data in a principled manner. In case of GP models, this concerns specifically the choice of the \emph{kernel}, 
which captures the covariance between function values at different inputs and is thus the core component to model prior knowledge about the function shape.
By choosing standard kernels, however, naive BO approaches do often not exploit this opportunity to improve learning performance. 

In this paper, we show how structural knowledge about the optimal control problem at hand can be leveraged for designing specific kernels in order to improve data efficiency in learning control.
For a first-order nonlinear quadratic optimal control problem, we propose two \emph{LQR kernels} that leverage the structure of the famous Linear Quadratic Regulator (LQR) problem given in form of an \emph{approximate linear model} of the \emph{true nonlinear dynamics}.  The proposed kernels 
leverage the structure of the LQR problem, while retaining the flexibility of nonparametric GP regression.
\subsubsection*{Contributions}
In detail, this paper makes the following contributions:
\begin{enumerate}[topsep=0pt]
\item We discuss how the structure of the well-known LQR problem can be leveraged for efficient learning of controllers for nonlinear systems.
\item This discussion leads to the proposal of two new kernels for GP regression in the context of learning control: the \emph{parametric} and \emph{nonparametric LQR kernels}.
\item The improved learning performance achieved with these kernels over a standard kernel is demonstrated through numerical simulations.
\end{enumerate}

\subsubsection*{Related work}
BO for learning controllers has recently been considered, for example in \cite{MaHeBoScTr16, MaBeHeScKrScTr17, BeScKr16, ScNgEbBiMaTo15, CaSePeDe15}, in different flavors.  While \cite{MaHeBoScTr16} and \cite{MaBeHeScKrScTr17} focus on data efficiency by maximizing the information gain in each iteration, \cite{BeScKr16} and \cite{ScNgEbBiMaTo15} propose methods for safe exploration. Different BO algorithms are compared in \cite{CaSePeDe15}.
These works 
all consider tuning of state-feedback controllers with a quadratic cost (possibly saturated) similar to the setting herein, but using standard kernels.
The design of customized kernels for GP regression has been considered before in the context of control and robotics for related problems.
A kernel for bipedal locomotion, which captures typical gait characteristics, is proposed in \cite{AnRaAt16}.  In \cite{MeEnHi16}, an impedance-based model is incorporated as prior knowledge 
for improved predictions in human-robot interaction.  For the problem of maximizing power generation in photovoltaic plants, the authors in \cite{AbBePoKr16} incorporate explicit basis functions about known power curves in the kernel.
In \cite{MaBeHeScKrScTr17}, a kernel is designed to model information from simulation and physical experiments, in order to leverage both sources of information for RL.

None of the above references considers the problem of incorporating the structure of the LQR problem
for improving data-efficiency in learning control with BO.

\subsubsection*{Outline}
We introduce the considered learning control problem in \sect \ref{sec:problem}, and summarize BO for RL in \sect \ref{sec:bayesianOpt}, together with the necessary background on GPs.  While the BO framework for learning control is introduced for general multivariate systems, we focus on the special case of a scalar problem thereafter and develop the LQR kernels in \sect \ref{sec:lqrkernel}.  Numerical results in \sect \ref{sec:simulations} illustrate the improved learning performance of the proposed kernels over a standard kernel.  The paper concludes with remarks in \sect \ref{sec:conclusion}.

\subsubsection*{Notation}
A Gaussian random variable $x$ with mean $m$ and variance $V$ is denoted by $x \sim \Nc(m,V)$.
The expected value of $x$ is denoted by $\E[x]$, while $\V[x_1, x_2]$ denotes the covariance of $x_1$ and $x_2$.  We also use $\V[x_1] := \V[x_1, x_1]$.

\section{Learning Control Problem} 	
\label{sec:problem}
We consider regulation of the nonlinear stochastic system
\begin{align}
x_{t+1} &= g(x_t, u_t, v_t)
\label{eq:sys}
\end{align}
with state $x_t \in \R^\dimx$, control input $u_t \in \R^\dimu$, and random noise $v_t \sim \Nc(0,V)$.
We assume that the state $x_t$ can be measured or otherwise estimated.
The system dynamics $g(\cdot)$ are unknown.  
The control objective is to find a state-feedback controller 
\begin{equation}
u_t = F x_t
\label{eq:stateFeedbackControl}
\end{equation}
with controller gain $F \in \R^{\dimu \times \dimx}$ such that the quadratic cost
\begin{equation}
J = \lim_{T \to \infty} \frac{1}{T} \E\Big[ \sum_{t=0}^{T-1} x_t^\transp Q x_t + u_t^\transp R u_t \Big]
\label{eq:cost}
\end{equation}
with symmetric weights $Q \geq 0$ and $R > 0$ is minimized.  A quadratic cost is a very common choice in optimal control \cite{AnMo07} expressing the designer's preference in the fundamental trade-off between control performance and control effort.

While the solution of the above problem is standard when the dynamics \eqref{eq:sys} are known and linear \cite{AnMo07}, here, we face the problem of optimal control under \emph{unknown} and \emph{nonlinear} dynamics.

To address this problem, we follow a direct RL approach \cite{ScAt10}, where the optimal controller is learned by directly evaluating the performance of candidate controllers $F_i$ on the real system \eqref{eq:sys} without learning an (intermediate) dynamics model.  
We employ Bayesian optimization (BO) as a popular approach for sequential global optimization, which leverages the information from previous trials in order to propose the next candidate controller $F_{i+1}$
so as to eventually find the optimum of \eqref{eq:cost}.
Because evaluations on the system are expensive, we seek to find the optimal controller with as few evaluations as possible.  To this end, we address herein how to leverage information from a linear model that approximates the true system \eqref{eq:sys}.
\section{Bayesian Optimization for Learning Control} 	
\label{sec:bayesianOpt}
We briefly introduce necessary background material on Gaussian processes (GPs) and Bayesian optimization (BO), before describing reinforcement learning with BO.

\subsection{Gaussian process regression}
\label{sec:GPR}
Let $\theta$ denote the free parameters of a feedback controller; that is, $\theta = \text{vec}(F)$ in \eqref{eq:stateFeedbackControl}, or any other parameterization.  The dependence of the cost function $J$ \eqref{eq:cost} on $\theta$ is unknown \emph{a priori} because of the lack of knowledge of the system \eqref{eq:sys}.  We use GP regression to approximate the function $J: \Theta \to \R$, $\theta \in \Theta$, from data (\ie noisy function evaluations)
\begin{equation}
J_i = J(\theta_i) + \varepsilon_i, \quad \varepsilon_i \sim \Nc(0,\sigma^2).
\label{eq:J_evaluation}
\end{equation}
Obtaining one data point $J_i$ corresponds to performing a closed-loop experiment on the true system \eqref{eq:sys} with controller $\theta_i$, recording the state-input trajectory, and computing the cost from these data (in case of \eqref{eq:cost} a finite approximation is used with a sufficiently long horizon $T$).

A GP can be defined as a probability distribution over the space of functions $J: \Theta \to \R$ whose restriction to any finite number of function values is jointly Gaussian \cite[p.~13]{RaWi06}.  
A GP is specified through its prior mean function $\mu: \Theta \to \R$ and covariance function $k: \Theta \times \Theta \to \R$. We write
\begin{align}
J(\theta) &\sim \GP(\mu(\theta), k(\theta, \theta^\prime) \quad \text{with} \nonumber \\
\mu(\theta) &= \E[J(\theta)] \nonumber \\
k(\theta,\theta^\prime) &= \V[J(\theta), J(\theta^\prime)] . \nonumber 
\end{align}
Hence, $\mu$ is the expected function value, and $k$ captures the covariance between any two function values and is used to model uncertainty.  The latter is also 
called \emph{kernel} and must satisfy the following property to give rise to a valid covariance function:
\begin{definition}[according to \cite{RaWi06}]
Let $k: \Theta \times \Theta \to \R$ be a function, and 
$K_N \in \R^{N \times N}$ be the symmetric matrix whose entries are computed as $K_{ij} = k(\theta_i, \theta_j)$ from the collection $\{ \theta_1, \theta_2, \dots, \theta_N\}$, $\theta_i \in \Theta$.  The function $k$ is called a \emph{positive semidefinite kernel} if $K$ is positive semidefinite for any finite collection $\{ \theta_1, \theta_2, \dots, \theta_N\}$.
\end{definition}
The matrix $K_N$ is called the \emph{Gram matrix}.

A kernel typically has its own parameters, called \emph{hyperparameters}, such as length scales or output variance. By choosing the prior mean and the kernel, the user can specify prior knowledge about the function such as expected shape, length scales, and smoothness properties.

In GP regression, learning a function amounts to predicting the (normal) distribution of function values $J(\theta^*)$ at arbitrary inputs $\theta^* \in \Theta$ based on previous evaluations.  Given $N$ data points $\Dc_N = \{\theta_i, J_i\}_{i=1}^N$, the posterior mean and variance of $J(\theta^*)$ can be stated in closed form as
\begin{align}
\E[J(\theta^*) | \Dc_N] & = \mu(\theta^*) - k_N^\transp(\theta^*) (K_N \!+\! \sigma^2 I)^{-1} \hat{J}_N \label{eq:GP_post_mean} \\
\V[J(\theta^*) | \Dc_N] & = k(\theta^*\!,\theta^*) \nonumber \\
 												& - k_N^\transp(\theta^*) (K_N \!+\! \sigma^2 I)^{-1} k_N(\theta^*) \label{eq:GP_post_var}
\end{align}
where $k_N(\theta^*) := [k(\theta^*\!,\theta_1), \dots, k(\theta^*\!,\theta_N)]^\transp$ and 
$\hat{J}_N := [J_1 - \mu(\theta_1), \dots, J_N - \mu(\theta_N)]^\transp$.

In addition to computing the posterior distribution, the data $\Dc_N$ can also be used to optimize hyperparameters in order to improve the model fit.  A popular approach is maximizing the marginal likelihood of the data, $p(\hat{J}_N | \theta_1, \dots, \theta_N)$, which is given in logarithmic form as \cite[p.~19]{RaWi06}
\begin{equation}
-\tfrac{1}{2}\hat{J}_N^\transp (K_N \!+\! \sigma^2 I)^{-1} \hat{J}_N 
-\tfrac{1}{2} \abs{K_N \!+\! \sigma^2 I}
-\tfrac{N}{2} \log 2\pi .
\label{eq:GP_marginalLikeli}
\end{equation}

\subsection{Bayesian optimization}
\label{sec:BO}
Bayesian optimization \cite{Mo89} denotes a class of global optimization methods, where a probabilistic description $p(J(\theta))$ of the latent objective function $J$ is used for data-efficient optimization. 
Here, the probabilistic description is given by the GP $J \sim \GP(\mu,k)$.  Given $p(J(\theta))$, 
a \emph{utility} $U[p(J(\theta))]$ is defined, which is maximized in each iteration to find the next evaluation point
\begin{equation}
\theta_{i+1} = \arg \max_\theta U[p(J(\theta))] .
\label{eq:BO_nextEval}
\end{equation}
Different BO algorithms primarily distinguish in how the selection of the next evaluation \eqref{eq:BO_nextEval} is formulated.
Popular BO algorithms include \emph{expected improvement} (EI) \cite{jones1998efficient}, \emph{probability of improvement} (PI) \cite{Ku64},
\emph{GP upper confidence bound} (GP-UCB) \cite{srinivas2009gaussian}, and \emph{entropy search} (ES) \cite{HeSc12}.

The method developed herein is agnostic to the type of BO algorithm used.  For the examples in \sect \ref{sec:simulations}, we employ EI, which uses 
\begin{equation}
U[p(J(\theta))] = \E[ \min \{0, J_\text{low} - J(\theta) \}]
\label{eq:EI}
\end{equation}
for the optimization in \eqref{eq:BO_nextEval}, where $J_\text{low}$ represents 
the lowest function from current data set $\Dc_N$. EI thus selects the next query point where the expected improvement upon $J_\text{low}$ (under the GP of $J$) is maximal.

\subsection{Reinforcement learning with Bayesian optimization}
Algorithm~\ref{alg:RLwBO} summarizes direct RL using Bayesian optimization. 
\begin{algorithm}[t!]
\caption{Learning control with Bayesian optimization
}
\begin{algorithmic}[1]
\State Specify objective function $J$ (\eg \eqref{eq:cost} with finite horizon $T$)
\State Specify GP prior (mean $\mu$, kernel $k$)
\State Initialize: controller $\theta_0$; data set $\Dc_0 = \emptyset$
\While{(not terminated)} 
\Comment \eg fixed number of experiments
\State perform closed-loop experiment with controller $\theta_i$
\State compute $J_i$ from experimental data $\{x_t, u_t\}_{t=0}^{T-1}$
\State add evaluation $\{\theta_i, J_i\}$ to data set $\Dc_i$
\State update GP posterior \eqref{eq:GP_post_mean}, \eqref{eq:GP_post_var}
\State [optional:] optimize hyperparameters (\eg maximize \eqref{eq:GP_marginalLikeli})
\State compute next controller $\theta_{i+1}$ via \eqref{eq:BO_nextEval}
\EndWhile
\State determine `best guess' $\theta_\text{bg}$ for the controller parameters, \eg minimum of posterior mean \eqref{eq:GP_post_mean}
\State \Return $\theta_\text{bg}$
\end{algorithmic}
\label{alg:RLwBO}
\end{algorithm}
This framework 
has been successfully applied in experimental settings to learn feedback controllers for the control problem of \sect \ref{sec:problem} and related scenarios.  Berkenkamp \etal \cite{BeScKr16} optimize a state-feedback controller \eqref{eq:stateFeedbackControl} for 
quadrotor trajectory tracking using a safe BO algorithm \cite{SuGoBuKr15}, which builds on GP-UCB.  Marco \etal \cite{MaHeBoScTr16} propose to parametrize the feedback gain $F$ as LQR policy \cite{TrMiDoDAn14}, whose optimal weights are learned using ES, which maximizes the information gain from each experiment.  This algorithm 
has been extended in \cite{MaBeHeScKrScTr17} to include information from different sources (such as simulation and physical experiments).  Both methods \cite{MaHeBoScTr16,MaBeHeScKrScTr17} were successfully applied to learn pole balancing controllers.  Calandra \etal \cite{CaSePeDe15} compare PI, EI, and GP-UCB for learning the parameters of a discrete event controller for a walking robot.

\section{Constructing LQR Kernels} 	
\label{sec:lqrkernel}
In this section, we present two kernels that are specifically designed for the control problem of \sect \ref{sec:problem} and incorporate approximate model knowledge in form of a linear approximation to \eqref{eq:sys}.  
Because for a linear system, the solution to the optimal control problem of \sect \ref{sec:problem} is the well-known LQR, 
we term these kernels \emph{LQR kernels}.  
The following derivations are presented for a first-order system \eqref{eq:sys}, where all variables are scalars ($\dimx = \dimu = 1$).  Small letters are used in place of the capital ones to emphasize scalar variables (\eg $v$ instead of $V$ in~\cref{eq:sys} and $f$ instead of $F$ in \eqref{eq:stateFeedbackControl}).  We consider minimization of the cost \eqref{eq:cost}, which is rewritten as
\begin{equation}
J = \lim_{T \to \infty} \frac{1}{T} \E\Big[ \sum_{t=0}^{T-1} qx_t^2  + ru_t^2 \Big].
\label{eq:cost_scalar}
\end{equation}

We start by considering a learning example with a standard kernel choice for the GP, which motivates why a specifically designed kernel can be desirable. 

\subsection{Problems with standard kernel}

The most common kernel choice in GP regression is arguably the squared exponential (SE) kernel \cite{RaWi06}
\begin{equation}
\kSE(\theta, \theta^\prime) = \sigma_\text{SE}^2 \exp\!\Big(-\frac{(\theta-\theta^\prime)^2}{2 \ell^2} \Big)
\label{eq:kernel_SE}
\end{equation}
with signal variance $\sigma_\text{SE}^2$ and length-scale $\ell$ as its hyperparameters.
Let us consider the problem of learning the cost function $J$ \eqref{eq:cost_scalar} via GP regression with SE kernel for the following linear example:
\begin{example} %
\label{ex:scalarLinear1}
Let \eqref{eq:sys} be given by $x_{t+1} = 0.9 x_t + u_t + v_t$ with $v_t \sim \Nc(0,1)$.
\end{example}

Figure \ref{fig:GPfit_SE} shows the prior GP and the posterior GP after obtaining four data points (\ie after four evaluations of controllers $f_i$).
A few issues are apparent from the posterior: (i) the kernel has problems with the different length scales of the function ($J$ is steep around $f=-0.1$, but rather flat in the center region); (ii) the GP does not generalize well to regions where no data has been seen (\eg around $f=-1.6$, where the posterior mean resorts to the prior); and (iii) the overall fit is not very good.  

Clearly, the fit will improve with more data, but, for efficient and fast learning of controllers, we are particularly interested in good fits from just few data points. Hence, we seek to improve the fitting properties of the GP by exploiting knowledge about the system \eqref{eq:sys} in terms of an approximate linear model.

\begin{figure}[tb]
\centering
\includegraphics[width=0.9\columnwidth]{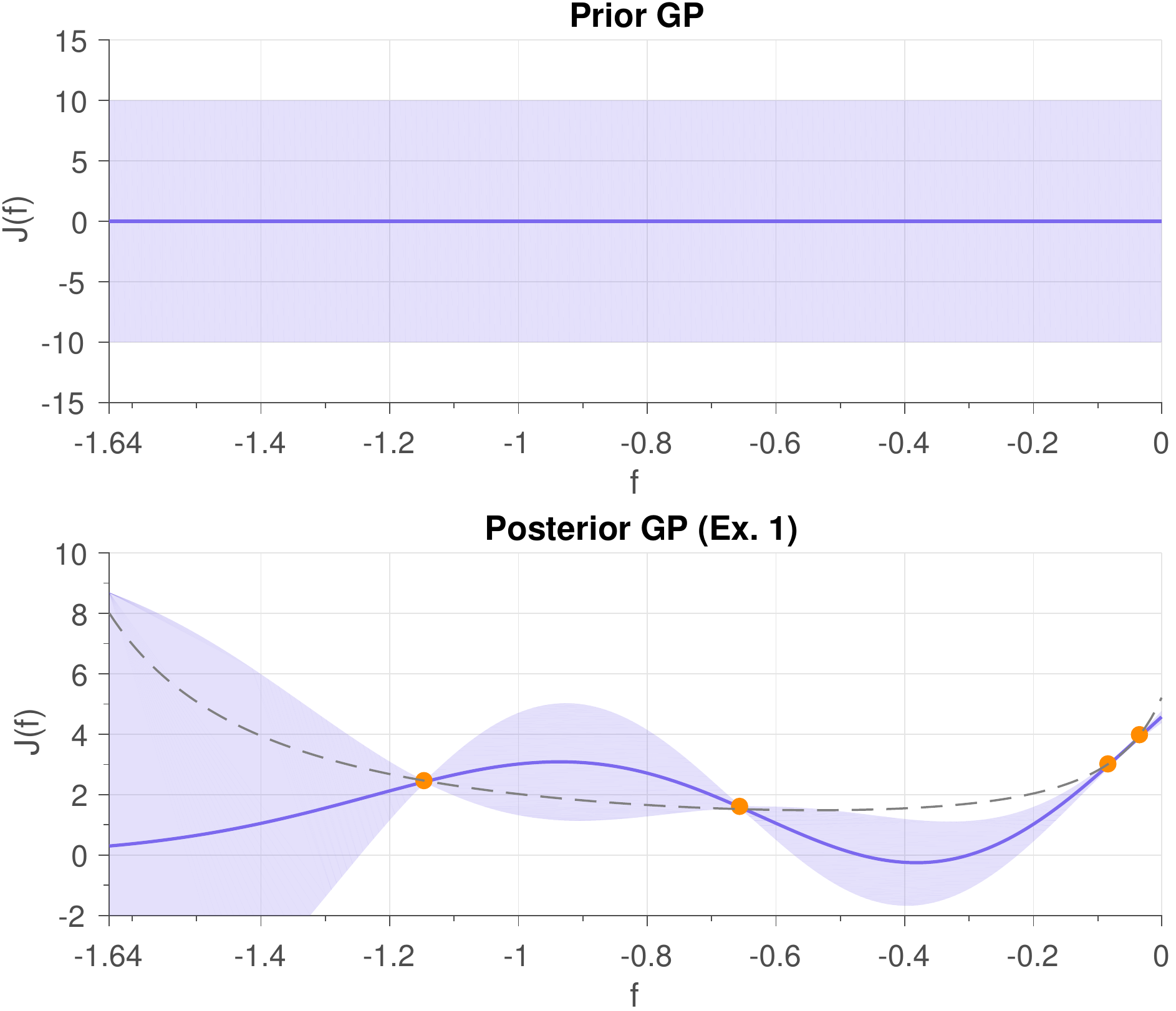}
\caption{Prior and posterior GPs of \Ex \ref{ex:scalarLinear1} using the squared exponential kernel $\kSE$ (hyperparameters: $\sigma_\text{SE}^2 = 25$, $\ell = 0.4$).  The thick line represents the GP mean, and the light blue area the GP variance (+/- two standard deviations).  The true function is shown in the bottom plot in dashed gray, and data points in orange.
}
\label{fig:GPfit_SE}
\end{figure}

\subsection{Incorporating prior knowledge}
In no practical situation, one has a perfect model of the system to be controlled.  At the same time, it is often possible to obtain a rough model, \eg from first principles modeling or some system identification procedure. 
Here, we assume 
 that an uncertain linear model 
\begin{align}
x_{t+1} &= a x_t + b u_t + v_t \label{eq:sys_lin} \\
&a \in [\alow, \aup], b \in [\blow, \bup] \label{eq:sys_lin_uncertainty}
\end{align}
is available as an approximation to \eqref{eq:sys}, \eg from li\-nearization of a first principles model
with (possibly) some uncertainty about the physical parameters.

In the following, we will consider controller gains $f$ such that the system \eqref{eq:sys_lin} is guaranteed stable for all parameters \eqref{eq:sys_lin_uncertainty}.  That is, we consider $f \in \setstable$ with
\begin{align}
\setstable := \{ f \!\in\! \R \, | & \, \abs{a+bf} < 1 \, \nonumber \\
& \forall a \!\in\! [\alow, \aup], b \!\in\! [\blow, \bup] \}.
\label{eq:f_stable_set}
\end{align}
This restriction makes sense, for example, in safety critical applications, where one wants to avoid the risk of trying an unstable controller based on the system knowledge available (\ie \eqref{eq:sys_lin}, \eqref{eq:sys_lin_uncertainty}).  Moreover, the restriction to $\setstable$ will ensure that subsequent calculations are well-defined.

If $a$ and $b$ were known, the functional dependence of the cost $J$ on the controller gain $f$ in \eqref{eq:stateFeedbackControl} for the linear system \eqref{eq:sys_lin} could be derived using standard control theory.
\begin{fact}
Consider the system \eqref{eq:sys_lin} with known parameters $a$ and $b$, and
let $\abs{a+bf}<1$. 
Then, the cost \eqref{eq:cost_scalar} is given by\footnote{In the notation $\phi_{(a,b)}(f)$, we omit the parametric dependence on $v$, $q$ and $r$ since $v$ is a multiplicative constant, which does not play a role in the later optimization, and we assume that $q$ and $r$ are fixed. } 
\begin{equation}
J = v \frac{q + rf^2}{1-(a+bf)^2} =: \phi_{(a,b)}(f).
\label{eq:J_fun_fab}
\end{equation}
\end{fact}
\begin{proof}
The controlled process $x_{t+1} = (a+bf) x_t + v_t$ is stable by assumption and thus converges to a stationary process with zero mean and variance $\E[x_t^2] = P$, where $P$ is the unique positive solution to $P - P(a+bf)^2 = v$, \cite[\The~3.1]{AnMo05}.
For stationary $x_t$, ~\cref{eq:cost_scalar} resolves into
\begin{equation*}
J = \E[ q x_t^2 + ru_t^2 ] 
= (q+rf^2) \E[ x_t^2 ] 
= (q+rf^2) P.
\end{equation*}
Equation \eqref{eq:J_fun_fab} then follows.
\end{proof}

If we are uncertain about $a$ and $b$, a collection of possible costs~\cref{eq:J_fun_fab} emerge from all the possible combinations of $a$ and $b$ within their ranges~\cref{eq:sys_lin_uncertainty}.
We assume this collection of costs to be explained by a Gaussian process $J_\text{LQR}$.
While any arbitrary choice of $a$ and $b$ in~\cref{eq:sys_lin} is only an approximation to~\cref{eq:sys}, it yields a cost~\cref{eq:J_fun_fab} that contains useful structural information, which can be leveraged for faster learning controllers from data.
In the next sections, we show how this prior knowledge can be exploited to construct LQR kernels.

\subsection{Parametric LQR kernel}
\label{ssec:par_LQRker}
A reasonable choice for $J_\text{LQR}(f)$ is
\begin{equation}
J_\text{LQR}(f) = w \, \phi_{(\bar{a},\bar{b})}(f), \quad w \sim \Nc(\bar{w},\sigma_w^2)
\label{eq:Jlqr_param}
\end{equation}
where $\bar{a} := (\alow+\aup)/2$ and $\bar{b} := (\blow+\bup)/2$ are the midpoints of the uncertainty intervals \eqref{eq:sys_lin_uncertainty}.

Equation \eqref{eq:Jlqr_param} is a standard parametric model with a single feature $\phi_{(\bar{a},\bar{b})}(f)$ and Gaussian prior.  
It is well-known (see e.g.\ \cite{RaWi06}) that $J_\text{LQR}(f)$ is a GP,
\begin{equation}
J_\text{LQR}(f) \sim \GP\big(0, \kLQR(f,f^\prime) \big)
\label{eq:GP_J_LQR1}
\end{equation}
where we have assumed $\bar{w}=0$, with kernel
\begin{align}
\kLQR(f,f^\prime) &:= \sigma_w^2 \phi_{(\bar{a},\bar{b})}(f)\phi_{(\bar{a},\bar{b})}(f^\prime) \nonumber \\
&= \sigma_\text{p}^2  \frac{  v^2(q + rf^2)(q + rf^{\prime2})}{(1-(\bar{a}+\bar{b}f)^2) (1-(\bar{a}+\bar{b}f^\prime)^2)}
\label{eq:kernel_lqr1}
\end{align} 
and hyperparameters $\sigma_\text{p}:=\sigma_w$, $\bar{a}$, and $\bar{b}$.  We refer to \eqref{eq:kernel_lqr1} as the \emph{parametric LQR kernel} because it captures the cost function for the linear system $(\bar{a}, \bar{b})$ with quadratic cost, and thus the structure of the LQR problem.

To illustrate the performance of the parametric LQR kernel $\kLQR$, we revisit \Ex \ref{ex:scalarLinear1} using this kernel instead of the SE kernel.
The top two graphs of \fig \ref{fig:GPfit_LQR1} show the prior and posterior GP for the same data points as in \fig \ref{fig:GPfit_SE} when using $\kLQR$ with hyperparameters  $\bar{a}=0.9$ and $\bar{b}=1$.  We see that the posterior fit from only four data points is almost perfect and much better than the one in \fig \ref{fig:GPfit_SE}.  This is, of course, not a big surprise because the hyperparameters match the true underlying system of \Ex \ref{ex:scalarLinear1} perfectly.  
The fitting performance deteriorates, however, if the hyperparameters are off.  This can be seen in the bottom graph of \fig \ref{fig:GPfit_LQR1}, which shows the posterior GP with the same hyperparameters, but for the system
\begin{example}%
\label{ex:scalarLinear2}
$x_{t+1} = 0.8 x_t + 0.9u_t + v_t$ with $v_t \sim \Nc(0,1)$.
\end{example}

\begin{figure}[tb]
\centering
\includegraphics[width=.9\columnwidth]{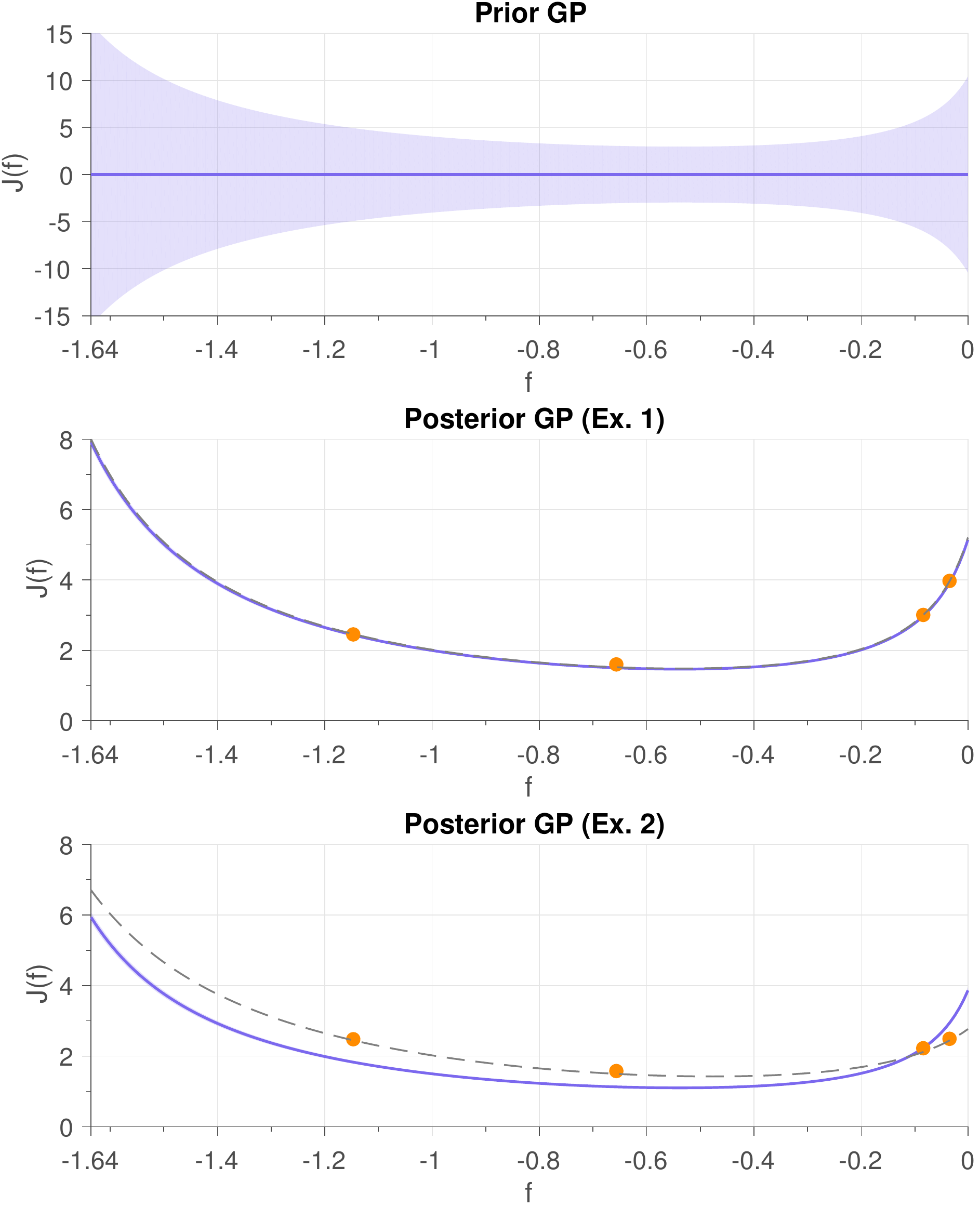}
\caption{GP fit using the parametric LQR kernel $\kLQR$ (hyperparameters: $\sigma_\text{p}^2=1$, $\bar{a}=0.9$, $\bar{b}=1$).  The color code is the same as in \fig \ref{fig:GPfit_SE}.  The hyperparameters are exact for \Ex \ref{ex:scalarLinear1}, while they are off by about 10\% for \Ex \ref{ex:scalarLinear2}.  
}
\label{fig:GPfit_LQR1}
\end{figure}

To improve the fit in this situation, one can employ hyperparameter optimization (see \sect \ref{sec:GPR}) to find improved parameters $\bar{a}$ and $\bar{b}$ that better explain the data.  The simulation results in \sect \ref{sec:simulations} will show that this can be a viable approach.  An alternative is to design more flexible and expressive kernels, which allow for fitting more general models.  This we discuss next.

\subsection{Nonparametric LQR kernel}
\label{ssec:nonpar_LQRker}
The kernel \eqref{eq:kernel_lqr1} captures the structure of the cost function for the LQR problem with one specific linear model $(\bar{a}, \bar{b})$.  
A straightforward way to increase the flexibility of the kernel in order to fit more general problems is to use 
$\nfeat$ basis functions (or features) $\phi_{(a_i, b_i)}$ corresponding to $\nfeat$ models $(a_1, b_1), (a_2, b_2), \dots, (a_\nfeat, b_\nfeat)$,
\begin{align}
J_\text{LQR}(f) &=  
\underbrace{
\big[
\phi_{(a_1, b_1)}(f) \,\,\,
\phi_{(a_2, b_2)}(f) \, \cdots \,
\phi_{(a_m, b_m)}(f) 
\big]
}_{=: \Phi^\transp(f)}
w \nonumber \\
&= \Phi^\transp(f) w
\label{eq:Jlqr_param_multi}
\end{align}
with $w \in \R^\nfeat$, $w \sim \Nc(\bar{w}, \Sigma_w)$.
The derivation of the corresponding kernel is analogous to \eqref{eq:kernel_lqr1} (see \cite[Sec.~2.7]{RaWi06}) and yields
\begin{equation}
\kLQRmulti(f,f^\prime) = \Phi^\transp\!(f) \, \Sigma_w \, \Phi(f^\prime) .
\label{eq:kernel_lqr1_multiFeatures}
\end{equation}

Same as  \eqref{eq:Jlqr_param}, the model \eqref{eq:Jlqr_param_multi} represents a parametric model for the LQR cost $J_\text{LQR}$.  That is, its flexibility is essentially limited to the number of explicit features $\phi_{(a_i, b_i)}$.
Employing powerful kernel techniques \cite{SmSc02}, 
the parametric model can be turned into a nonparametric one, which includes an infinite number of features while retaining finite computational complexity.  
The key idea is to consider the kernel \eqref{eq:kernel_lqr1_multiFeatures} in the limit 
of infinitely many features corresponding to models $a \in [\alow, \aup]$ and $b \in [\blow, \bup]$.  The derivation follows ideas similar to how the standard SE kernel can be derived from basic features, \cite[p.~84]{RaWi06}.

Consider the partitions of $[\alow, \aup]$ and $[\blow, \bup]$ into $\nfeat$ equidistant intervals, and let $\{a_i\}_{1:\nfeat}$ and $\{b_i\}_{1:\nfeat}$ be the lower (or upper) interval limits.  Consider the model \eqref{eq:Jlqr_param_multi} with feature vector 
\begin{equation*}
\Phi^\transp(f) = 
\begin{bmatrix}
\phi_{(a_1, b_1)}(f)  \cdots  \phi_{(a_i, b_j)}(f)  \cdots  \phi_{(a_m, b_m)}(f)
\end{bmatrix}
\end{equation*}
which includes all combinations $\phi_{(a_i, b_j)}$ for $i,j \in \{1,\dots,\nfeat\}$, and the parametric prior $\bar{w} = 0$ and $\Sigma_w = \frac{\sigma_\text{n}^2 (\aup-\alow)(\bup-\blow)}{\nfeat^2} I$ for some $\sigma_\text{n} \in \R$.  The kernel \eqref{eq:kernel_lqr1_multiFeatures} then becomes
\begin{align}
\kLQRmulti(f,f^\prime) &= \frac{\sigma_\text{n}^2 (\aup-\alow)(\bup-\blow)}{\nfeat^2} \nonumber \\
&\phantom{=}\times  \sum_{j=1}^{\nfeat} \sum_{i=1}^{\nfeat} \phi_{(a_i, b_j)}(f) \, \phi_{(a_i, b_j)}(f^\prime).
\label{eq:derivation_lqr_finite_sum}
\end{align} 
Since $\phi_{(a_i, b_j)}$ is continuous on $\setstable$, the finite sum \eqref{eq:derivation_lqr_finite_sum} converges to the Riemann integral in the  limits as $\nfeat \to \infty$.  We can thus define the \emph{nonparametric LQR kernel} 
\begin{align}
\kLQRtwo(f,f^\prime) &= \lim_{\nfeat \to \infty} \kLQRmulti(f,f^\prime) \nonumber \\
&= \sigma_\text{n}^2 \!\! \int_{\blow}^{\bup} \!\! \int_{\alow}^{\aup} \!\! \phi_{(a, b)}(f)  \phi_{(a, b)}(f^\prime) \, da \, db 
\label{eq:kernel_lqr2}
\end{align}
for $f,f^\prime \in \setstable$ with
the signal variance $\sigma_\text{n}^2$ and the integration boundaries $\alow$, $\aup$, $\blow$, and $\bup$ as hyperparameters.
While the kernel in \eqref{eq:kernel_lqr2} represents the structure of the cost function \eqref{eq:J_fun_fab} for an infinite number of models (all $a \in [\alow, \aup]$ and $b \in [\blow, \bup]$), its computation is finite consisting of solving the integral in \eqref{eq:kernel_lqr2}.  By contrast,  the computational complexity of the parametric kernels \eqref{eq:kernel_lqr1_multiFeatures} and \eqref{eq:derivation_lqr_finite_sum} grows with the number of features $\nfeat$.  We prove that the kernel \eqref{eq:kernel_lqr2} is indeed a valid covariance function. 
\begin{proposition}
$\kLQRtwo(f,f^\prime)$ is positive semidefinite for all $f, f^\prime \in \setstable$. 
\end{proposition}
\begin{proof}
Take any collection $\{\theta_1, \theta_2, \dots, \theta_N\}$ and any $z \in \R^N$. Let $K_N$ be the Gram matrix for the kernel $\kLQRtwo$. Then
\begin{align*}
z^\transp K_N z
&= \sum_{j=1}^N z_j \Big( \sum_{i=1}^N z_i \kLQRtwo(\theta_i, \theta_j) \Big) \\
&= \sigma_\text{n}^2 \sum_{j=1}^N z_j  \sum_{i=1}^N z_i \! \!   \int\limits_{\blow}^{\bup}  \int\limits_{\alow}^{\aup} \!\! \phi_{(a, b)}(\theta_i)  \phi_{(a, b)}(\theta_j) \, da \, db \\
&= \sigma_\text{n}^2 \int\limits_{\blow}^{\bup}  \int\limits_{\alow}^{\aup}  \sum_{j=1}^N z_j \phi_{(a, b)}(\theta_j) \sum_{i=1}^N z_i \phi_{(a, b)}(\theta_i) \, da \, db  \\
&= \sigma_\text{n}^2 \int\limits_{\blow}^{\bup}  \int\limits_{\alow}^{\aup} \Big(\sum_{i=1}^N z_i \phi_{(a, b)}(\theta_i) \Big)^2 \, da \, db \geq 0
\end{align*}
which completes the proof.
\end{proof}

The above derivation corresponds to the \emph{kernel trick} \cite{RaWi06,SmSc02}, 
which is a core idea of kernel methods and behind many powerful learning algorithms. In essence, the kernel trick means to write a learning algorithm solely in terms of inner products of features and replacing those by a kernel\footnote{All computations for the GP regression, \ie equations \eqref{eq:GP_post_mean}, \eqref{eq:GP_post_var}, and \eqref{eq:GP_marginalLikeli}, are expressed solely in terms of the kernel $k$ (and the mean $\mu$).}.
In particular, this allows for considering an infinite number of features, while retaining finite computation.  

Figure \ref{fig:GPfit_LQR2} shows the prior and posterior GP for the nonparametric LQR kernel \eqref{eq:kernel_lqr2}.  Because the kernel is more flexible than \eqref{eq:kernel_lqr1},
it can fit the cost functions for the two different models of \Ex \ref{ex:scalarLinear1} and \Ex \ref{ex:scalarLinear2} well.

\begin{figure}[tb]
\centering
\includegraphics[width=.9\columnwidth]{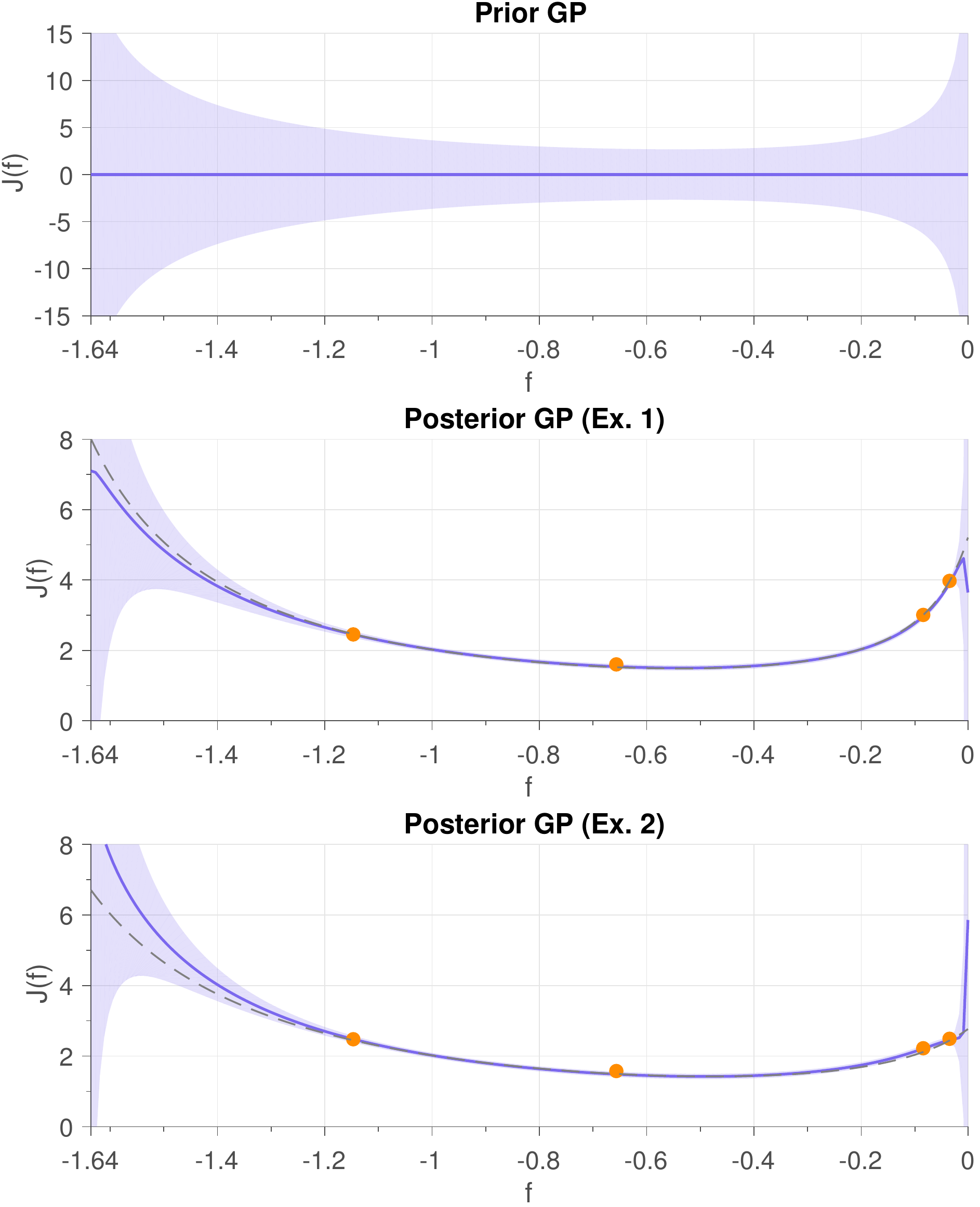}
\caption{GP fit using the nonparametric LQR kernel $\kLQRtwo$ (hyperparameters: $\sigma_\text{n}^2 = 20$, $\alow=0.8$, $\aup=1.0$, $\blow=0.9$, and $\bup=1.1$).  Colors are the same as in \fig \ref{fig:GPfit_SE}. Both examples are fitted well.
}
\label{fig:GPfit_LQR2}
\end{figure}

\subsection{A combined kernel}
System \eqref{eq:sys_lin} is an approximation to the true system \eqref{eq:sys}. It is thus not desirable to fully commit to the model for solving the optimal control problem. This would mean to directly minimize \eqref{eq:J_fun_fab} (which would result in the well-known LQR solution).  On the other hand, the linear problem contains information about the structure of the optimization problem, which shall be useful also for optimization of the true nonlinear system \eqref{eq:sys}, as long as we believe \eqref{eq:sys_lin} to be a reasonable approximation thereof.  In other words, 
we can expect the true cost function $J$ for the nonlinear problem to bear some similarity to \eqref{eq:J_fun_fab}.
We model the cost function \eqref{eq:cost} for the nonlinear system \eqref{eq:sys} as 
being composed of a part that stems from the approximation as LQR problem and an error term,
\begin{equation}
J(f) = J_\text{LQR}(f) + J_\Delta(f).
\label{eq:J_decomposition}
\end{equation}
The term $J_\Delta(f)$ captures the error in the model that stems from the fact that the true problem is nonlinear.  We model it as a standard GP, \eg employing the SE kernel \eqref{eq:kernel_SE}: $J_\Delta(f) \sim \GP(0,\kSE(f,f^\prime))$. 
We can model $J_\text{LQR}(f)$ as a GP~\cref{eq:GP_J_LQR1} using either the parametric LQR kernel~\cref{eq:kernel_lqr1} or the nonparametric~\cref{eq:kernel_lqr2} LQR kernel. Then, since the sum of two independent Gaussians is also Gaussian, it follows from \eqref{eq:J_decomposition} that also 
$J$ is a GP (see \cite[Sec.~2.7]{RaWi06}), with  
\begin{equation*}
J(f) \sim \GP\big(0, k_\text{LQR}(f,f^\prime) +  \kSE(f,f^\prime) \big).
\label{eq:GP_J_1}
\end{equation*}
where $k_\text{LQR}$ can be replaced by~\cref{eq:kernel_lqr1} or~\cref{eq:kernel_lqr2}.
By choosing the hyperparameters of the kernels, the designer can express how much he or she trusts the LQR versus the SE model.  For example, $\sigma_\text{SE} = 0$ means to fully rely on the LQR kernel.  %
\section{Simulations} 	
\label{sec:simulations}

In this section, 
we show statistical comparisons of the LQR kernels proposed in \cref{sec:lqrkernel} against the commonly used SE kernel, in two different settings.
In the first setting, we evaluate the performance of each kernel in the context of GP regression. Specifically, we quantify the mismatch between the GP posterior mean, computed from a set of random evaluations, and the underlying cost function.
In the second setting, we evaluate each kernel in the context of BO by comparing the learned minimum to the true global minimum.

The GP regression and BO experiments are presented in \cref{ssec:GPfit} and \cref{sec:BO_lin}, respectively, considering a linear system \eqref{eq:sys}. In addition, we also evaluate the BO setting for a nonlinear system in \cref{sec:BO_nonlin}.

\subsection{Experimental choices}
\label{ssec:exp_choices}
For the simulations in~\cref{ssec:GPfit} and \cref{sec:BO_lin}, we consider the true system~\cref{eq:sys} to be linear (i.e., \cref{eq:sys_lin}) with uncertain parameters 
$a \in \left[ 0.8, 1.0 \right]$ and $b \in \left[ 0.9, 1.1 \right]$.
We consider the optimal control problem as in \sect \ref{sec:problem} with $q=r=v=1$.  Feedback controllers \cref{eq:stateFeedbackControl} are considered to be in the range \eqref{eq:f_stable_set}, $\domf = \left[ -1.64, -0.001 \right]$.

For each controller $f$, the corresponding infinite-horizon LQR cost $J(f)$ is given by \cref{eq:J_fun_fab}.
In practice, only finite-horizon simulations can be realized. Therefore, the outcome of an experiment is noisy, as modeled in \eqref{eq:J_evaluation}, with $\sigma=0.05$.
In each simulation, a different linear model {\ab } is obtained by uniformly sampling the ranges above, which yields a different underlying cost function~\cref{eq:J_fun_fab}.
Each simulation is repeated for four different kernels:
\begin{itemize}
	\item \textbf{SE kernel}: As described in \cref{eq:kernel_SE}.
	\item \textbf{LQR kernel I}: Parametric LQR kernel \cref{eq:kernel_lqr1} with fixed parameters $(\bar{a}, \bar{b})$ (midpoints of uncertainty intervals).
	\item \textbf{LQR kernel II}: Parametric LQR kernel \cref{eq:kernel_lqr1}, with $(\bar{a}, \bar{b})$ optimized from evaluations by maximizing \cref{eq:GP_marginalLikeli}.
	\item \textbf{LQR kernel III}: Nonparametric LQR kernel \cref{eq:kernel_lqr2}.
\end{itemize}

The parametric LQR kernel \cref{eq:kernel_lqr1} is constructed taking the middle points of the uncertainty intervals, i.e., $\bar{a}=0.9$ and $\bar{b}=1$.
For the nonparametric LQR kernel \cref{eq:kernel_lqr2}, the intervals serve as integration domains.
The length scale of the SE kernel is computed as one fifth of the input domain, i.e., $\ell = 0.327$.
The signal variance of the parametric and nonparametric LQR kernel, i.e., $\sigma_\text{p}^2$ and $\sigma_\text{n}^2$, are normalized such that $\kLQR(\bar{f},\bar{f})=\kLQRtwo(\bar{f},\bar{f})=10$, where $\bar{f} = -0.82$ is the midpoint of $\domf$. Since the variance of these two kernels grow fast toward the corners of the domain, we set for the SE kernel $\sigma_\text{SE}^2=50$, for a fair comparison.

\subsection{GP regression setting}
\label{ssec:GPfit}
For this statistical comparison, we run 1000 simulations. In each simulation, we compute the GP posterior conditioned on two evaluations randomly sampled from the underlying cost function. We assess the quality of the regression by computing the root mean squared error (RMSE) between the true cost function and the GP posterior mean, both computed on a grid of 100 points over $\domf$.

\begin{figure}[tb]
\centering
\includegraphics[width=1\columnwidth]{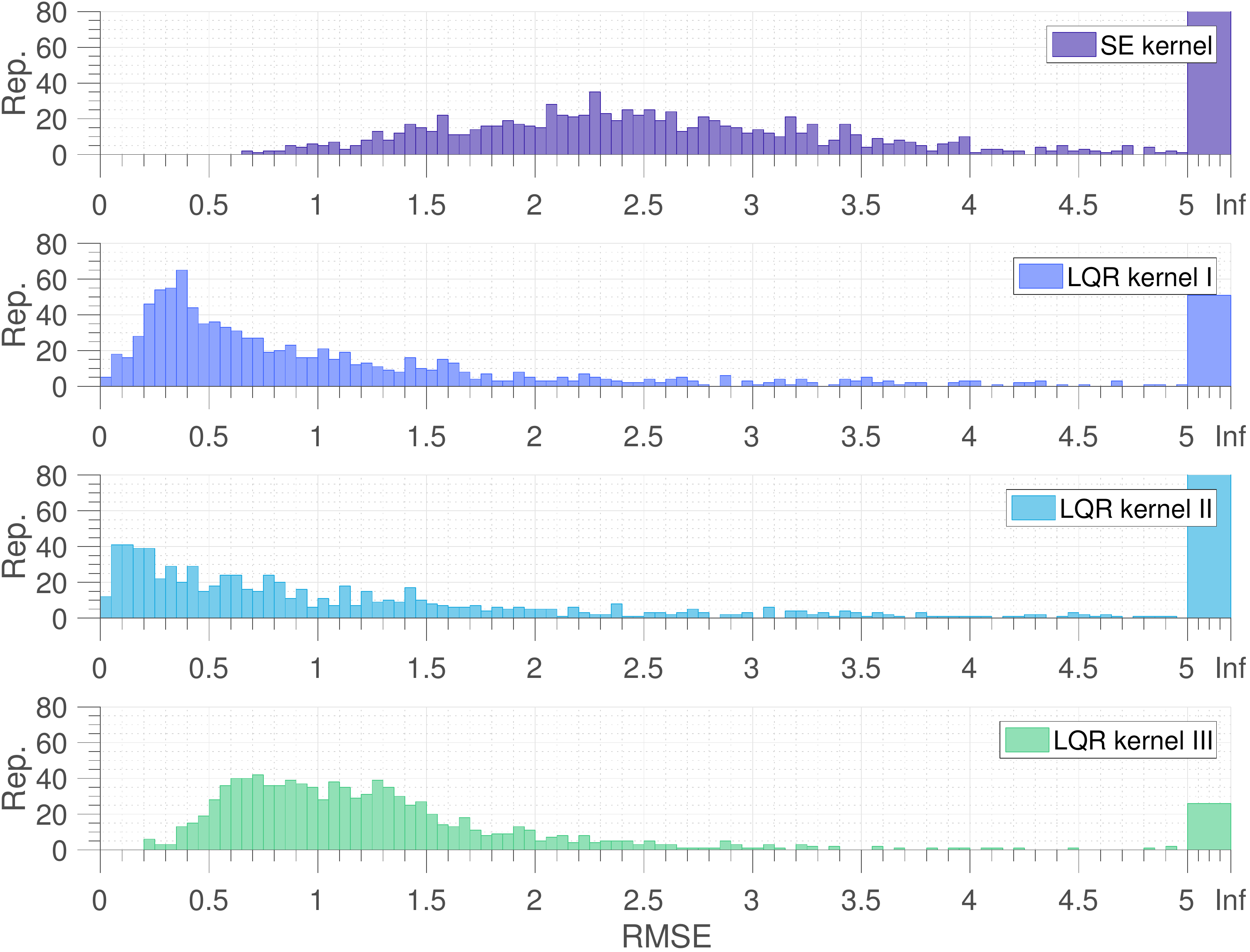}
\caption{Histograms of the root mean squared error (RMSE) between the true cost function and the GP posterior mean conditioned on two evaluations. The histograms are computed out of 1000 simulations.}
\label{fig:sec5_hist_fit_Neval2}
\end{figure}

\cref{fig:sec5_hist_fit_Neval2} shows the histograms of the RMSE obtained with the different kernels.
The LQR kernels clearly outperform the SE kernel in these experiments because they contain structural knowledge about the true cost, which contributes to a better GP fit, even with only two data points.
The nonparametric kernel makes good predictions because it inherently contains information about all possible functions within their uncertainty ranges of \ab. However, poorly specified integration bounds will decrease its performance. The RMSE statistical analysis confirms this when the integration intervals on {\ab } are a 50$\%$ larger than their uncertainties.
The parametric kernel with fixed hyperparameters ($\bar{a}, \bar{b}$) has a significant number of outliers since, in many cases, the data is queried from a cost function whose sampled {\ab } are far away from the ones of the kernel. The parametric kernel with hyperparameter optimization also leads to a better fit than SE, but has most outliers (about $75 \%$ of the simulations) since hyperparameter optimization is not reliable with just two data points.

\begin{table}[b]
	\setlength{\tabcolsep}{4pt}
	\centering
	\caption{RMSE averaged over 1000 simulations.}
	\label{tab:RMSE}
	\begin{tabular}{c|ccccc}
    $N$ & SE ker. & SE ker. (*) & LQR ker. I & LQR ker. II & LQR ker. III \\
		\hline
		1  & 2.76 (0.74) & 2.39 (0.81) & \textbf{1.01} (0.95) & 1.98 (1.29) & 1.31 (0.71) \\
		2  & 2.49 (0.85) & 2.42 (0.93) & \textbf{1.02} (0.96) & 1.09 (1.07) & 1.22 (0.67) \\
		5  & 1.83 (1.02) & 1.31 (0.95) & 1.10 (1.05) & \textbf{0.45} (0.70) & 1.20 (0.76) \\
		10 & 1.13 (0.99) & 0.98 (0.86) & 0.98 (0.96) & \textbf{0.20} (0.46) & 1.11 (0.74) \\
	\end{tabular}
\end{table}

We have repeated these experiments with $N=1$, $5$, and $10$ evaluations. \cref{tab:RMSE} shows the averaged RMSE over 1000 simulations for each kernel, and the corresponding standard deviation (in parentheses). The the outliers (i.e., any RMSE above 5) were excluded from these computations. 
In general, we see that the LQR kernel optimized from data performs better than the others for more than 2 evaluations.
For a fair comparison, we also include the SE kernel with hyperparameter optimization in the table (marked with an asterisk).
Because it performs similar to the SE kernel, and it does not improve upon the LQR kernels, we leave it out of the discussion for the rest of the paper.

\emph{Remark}:
The hyperparameters {\ab } of the parametric LQR kernel are optimized from data by maximizing the marginal likelihood \cref{eq:GP_marginalLikeli}. In a sense, optimizing the hyperparameters {\ab } of the LQR kernel from data can be considered similar to doing system identification on the linear system {\ab } using \cref{eq:GP_marginalLikeli} as performance metric. 
\subsection{BO setting}
\label{sec:BO_lin}
In this section, we evaluate the performance of each kernel in the context of BO. 
For each BO run, the first evaluation is decided randomly within the range of controllers $\domf$. Subsequent evaluations are acquired using the expected improvement (EI) method \cref{eq:BO_nextEval}, \cref{eq:EI}. We stop the exploration after three evaluations and compute the instantaneous regret (i.e., the absolute error between the true minimum and the minimum of the GP posterior mean) as the outcome of each experiment. 

\begin{figure}[tb]
\centering
\includegraphics[width=1\columnwidth]{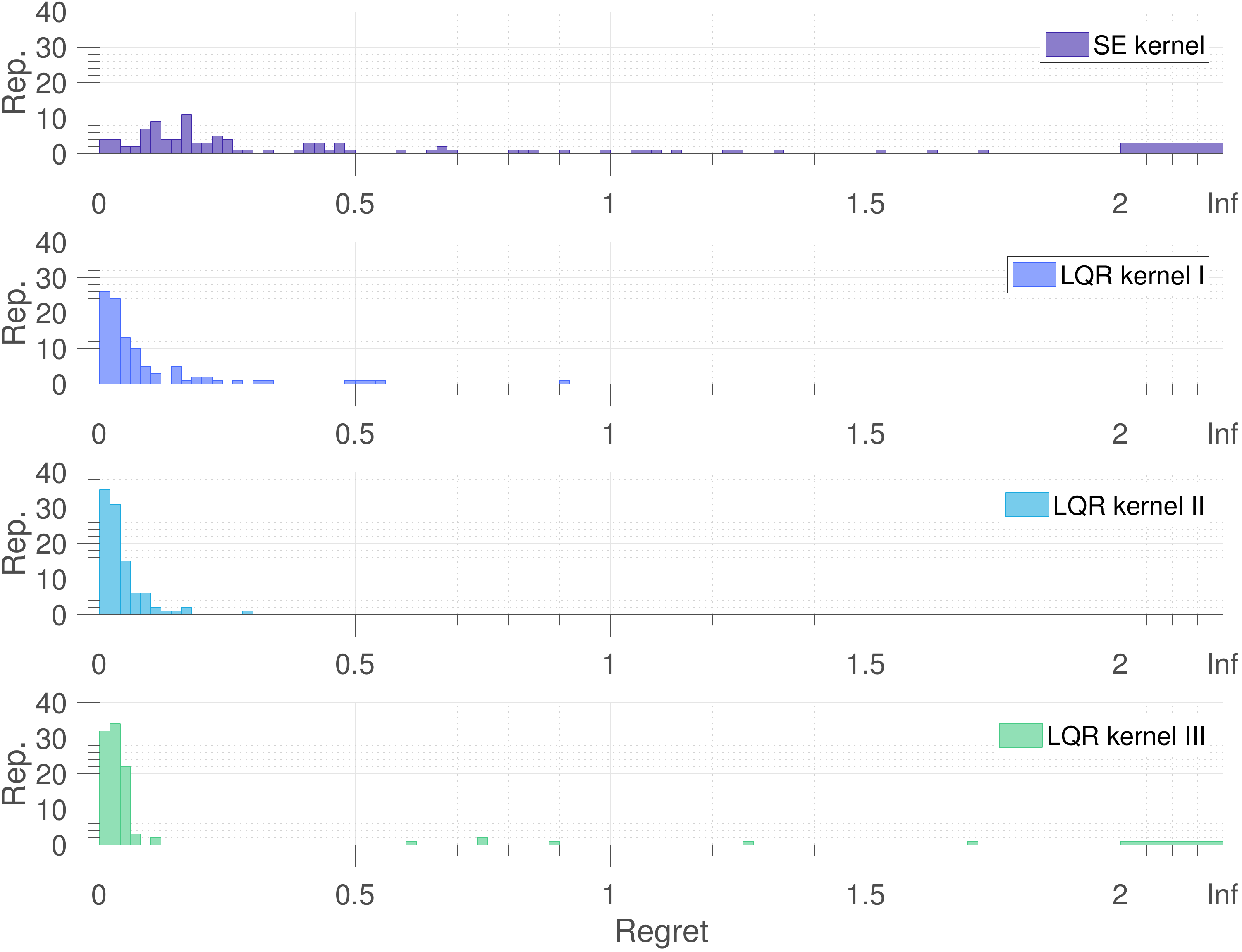}
\caption{Histogram of the regret incurred by stopping BO after three evaluations. The system considered is the linear system \cref{eq:sys_lin}. The histogram is computed over 100 BO runs.}
\label{fig:sec5_hist_BO_Neval2}
\end{figure}

\cref{fig:sec5_hist_BO_Neval2} shows the histogram of the regret for each kernel over 100 BO runs. The LQR kernels consistently outperform the SE kernel. The nonparametric kernel shows some outliers because in some cases the GP posterior mean grows large toward negative values, and so does its minimum. This is a wrong prediction of the underlying cost, defined positive. 
However, this minor issue can easily be detected, 
and the optimization procedure continued with a randomly sampled controller.
\subsection{BO setting for a nonlinear system}
\label{sec:BO_nonlin}
In this section, we use the same BO setting as in \cref{sec:BO_lin}, but consider now a nonlinear system \eqref{eq:sys}, namely
\begin{equation}
	x_{t+1} = \tilde{a}\sin\left( x_t \right) + \tilde{b} u_t + v_t
	\label{eq:sys_nonlinear}
\end{equation}
with $v_t \sim \Nc(0,1)$ and uncertain parameters $\tilde{a} \in \left[ 0.9, 1.1 \right]$, $\tilde{b} \in \left[ 0.9, 1.1 \right]$.
We control this system from $x_0=1$ to zero, using the same controller structure as the one described in \cref{ssec:exp_choices}. In this case, the considered range of controllers is $\domf = \left[ -1.57, -0.27 \right]$, which corresponds to~\cref{eq:f_stable_set}, reduced by a $20\%$. 
The LQR kernels are built up using the linearized version of \eqref{eq:sys_nonlinear} around the zero equilibrium point, \ie $x_{t+1} = \tilde{a} x_t + \tilde{b} u_t + v_t$.
\cref{tab:Regret_nonlin} shows the regret average and standard deviation. As can be seen, the LQR kernels perform better than the SE kernel.

\begin{table}[t]
	\centering
	\caption{Regret averaged over 100 BO runs for a nonlinear system}
	\label{tab:Regret_nonlin}
	\begin{tabular}{c|cccc}
		$N$ & SE kernel & LQR kernel I & LQR kernel II & LQR kernel III \\
		\hline
		2 & 1.34 (0.33) & \textbf{0.30} (0.21) & 0.32 (0.20) & 0.35 (0.18) \\
		3 & 0.49 (0.39) & 0.33 (0.38) & \textbf{0.31} (0.19) & 0.32 (0.18) \\
		4 & 0.35 (0.27) & 0.36 (0.44) & 0.32 (0.18) & \textbf{0.32} (0.17) \\
		5 & 0.36 (0.21) & \textbf{0.31} (0.36) & 0.32 (0.20) & 0.32 (0.18) \\
	\end{tabular}
\end{table}

\section{Concluding Remarks} 	
\label{sec:conclusion}
In this paper, we discussed how prior knowledge about the structure of an optimal control problem can be leveraged for data-efficient learning control.  Specifically, for a nonlinear quadratic optimal control problem, we showed how an uncertain linear model approximating the true nonlinear dynamics can be exploited in a Bayesian setting.  This led to the proposal of two novel kernels, a parametric and a nonparametric version of an LQR kernel, which incorporate the structure of the well-known LQR problem as prior knowledge.  Numerical simulations herein demonstrate improved data efficiency over standard kernels, \ie good controllers are learned from fewer experiments.  
We hope that the discussion and analysis herein also motivate further development of kernels tailored for other learning control problems.

Approaching the nonlinear quadratic optimal control problem presented herein with pure model-based methods can lead to 
superior performance for very accurate models compared to the proposed data-based approach.
However, this paper shares the motivation of \cite{MaHeBoScTr16}, which proposes a data-based approach when only poor models are available, and extends it by incorporating the LQR structure into the kernel.

The results herein are preliminary in the sense that the LQR kernels are developed for a first-order system.  While this is the natural first step, future work will concern the extension of the ideas and derivations to multivariate systems in order to develop this into a powerful framework for learning control in practice.  
Moreover, we plan to validate the benefit of the new kernels in experiments on more realistic nonlinear examples or physical hardware such as those in \cite{MaHeBoScTr16} and \cite{MaBeHeScKrScTr17}, settings that pose challenging issues, like imperfect states measurement, among others.

\bibliographystyle{IEEEtran}
\bibliography{LqrKernel-nocomm}

\end{document}